\documentclass[11pt]{article}

\usepackage[margin=1in]{geometry}
\usepackage{amsmath, amssymb, amsthm}
\usepackage{graphicx}
\usepackage{booktabs}
\usepackage{hyperref}
\usepackage{enumitem}
\usepackage{bm}
\usepackage[numbers]{natbib}

\hypersetup{
	colorlinks=true,
	linkcolor=blue,
	citecolor=blue,
	urlcolor=blue
}

\newtheorem{definition}{Definition}
\newtheorem{proposition}{Proposition}

\title{
	\bfseries
	Complementary Characterization of Agent-Based Models via\\
	Computational Mechanics and Diffusion Models
}

\author{
	Roberto Garrone\\
	\small University of Milano-Bicocca\\
	\small \texttt{roberto.garrone@unimib.it}
}

\date{\today}

\begin{document}
	\maketitle
	
	\begin{abstract}
	This article extends the preprint \emph{Characterizing Agent-Based Model Dynamics via 
		$\epsilon$-Machines and Kolmogorov-Style Complexity}
	by introducing diffusion models as orthogonal and complementary 
	tools for characterizing the output of agent-based models (ABMs).
	Where $\epsilon$-machines capture the predictive temporal structure 
	and intrinsic computation of ABM-generated time series, diffusion 
	models characterize high-dimensional cross-sectional distributions, 
	learn underlying data manifolds, and enable synthetic generation of 
	plausible population-level outcomes. 
	
	We provide a formal analysis demonstrating that the two approaches 
	operate on distinct mathematical domains---processes vs.\ distributions---
	and show that their combination yields a two-axis representation of ABM 
	behavior based on temporal organization and distributional geometry. 
	To our knowledge, this is the first framework to integrate 
	computational mechanics with score-based generative modeling for 
	the structural analysis of ABM outputs, thereby situating ABM 
	characterization within the broader landscape of modern machine-learning 
	methods for density estimation and intrinsic computation.
	
	The framework is validated using the same elder--caregiver ABM dataset 
	introduced in the companion paper, and we provide precise definitions 
	and propositions formalizing the mathematical complementarity between 
	$\epsilon$-machines and diffusion models. 
	This establishes a principled methodology for jointly analyzing temporal 
	predictability and high-dimensional distributional structure in complex 
	simulation models.

	\end{abstract}
	
\section{Introduction}

Agent-based models (ABMs) generate complex outputs composed of 
both temporal trajectories (e.g.\ agent states evolving over time) 
and cross-sectional or high-dimensional snapshots 
(e.g.\ populations of agents with multivariate attributes).  
In a companion preprint \citep{garrone2025epsilonmachine}, 
we analyzed ABM temporal outputs using computational mechanics 
\citep{crutchfield1989inferring,shalizi2001computational} 
and $\epsilon$-machines \citep{crutchfield2012between}, 
demonstrating substantial heterogeneity in intrinsic computation 
across spatial and socioeconomic variables.

However, ABMs also produce high-dimensional, non-sequential outputs 
for which $\epsilon$-machines are not suitable. 
Diffusion models---a class of score-based generative models 
\citep{sohldickstein2015deep,ho2020denoising,song2019generative,song2021maximum}---
excel at representing complex, high-dimensional distributions.  
This manuscript demonstrates that $\epsilon$-machines and diffusion 
models provide two \emph{complementary} and non-overlapping 
characterizations of ABM output, enabling a two-dimensional 
representation of model behavior: temporal structure versus 
distributional geometry. Formally, we distinguish between the process domain, analyzed via a mapping 
$\Phi : P \to M$, and the distribution domain, analyzed via a mapping 
$\Psi : D \to G$, a distinction developed in Section~3.

To our knowledge, this is the first work to integrate 
computational mechanics and diffusion-based generative modeling 
into a unified analytic framework for ABM characterization.  
By treating temporal and distributional structures as 
distinct informational domains, we introduce a principled 
methodology for jointly analyzing intrinsic computation and 
population-level geometry, providing a perspective not available 
in existing ABM, complexity-science, or machine-learning approaches.

We provide mathematical definitions and propositions to 
formalize this complementarity, apply the framework to the 
elder--caregiver ABM introduced in the companion paper 
\citep{garrone2025epsilonmachine}, and outline 
a unified methodology for combining the two analytic tools.

	\section{Background}
	
	\subsection{ABM Outputs: Temporal and Distributional Components}
	
	Let an ABM generate:
	\begin{itemize}[noitemsep]
		\item temporal sequences 
		$X^{(i)}_{0:T} = (X^{(i)}_0, \dots, X^{(i)}_T)$ 
		for agents $i = 1,\dots,N$, and
		\item cross-sectional population snapshots 
		$Y_t = (Y_{t,1}, \dots, Y_{t,N})$ 
		with each agent represented by a vector 
		$Y_{t,i} \in \mathbb{R}^d$.
	\end{itemize}
	
	These two object types—sequences and population vectors—are governed by fundamentally 
	different statistical structures and therefore require distinct analytic tools.

	\subsection{$\epsilon$-Machines (Computational Mechanics)}
		Computational mechanics characterizes a stationary stochastic process by reconstructing the
	minimal predictive architecture compatible with its observable dynamics. Given a sequence
	$\{X_t\}$ generated by an ABM, two pasts $x_{:t}$ and $x'_{:t}$ are said to be predictively
	equivalent when they yield identical conditional distributions over all possible futures:
	\[
	\mathbb{P}(X_{t:\infty} \mid X_{:t} = x_{:t})
	=
	\mathbb{P}(X_{t:\infty} \mid X_{:t} = x'_{:t}) .
	\]
	Each equivalence class defines a \emph{causal state}, and the resulting collection of causal states,
	together with the transition structure induced by observed symbol sequences, forms the
	$\epsilon$-machine: the minimal unifilar hidden Markov model that captures the process’s
	predictive organization.
	Causal states 
	\citep{crutchfield1989inferring,shalizi2001computational} 
	are defined as equivalence classes:
	\[
	\epsilon(x_{:t}) = \{x'_{:t} : 
	P(X_{t:\infty} \mid x'_{:t}) = P(X_{t:\infty} \mid x_{:t})\}.
	\]
	
	Three standard informational quantities summarize the structure encoded in the
	$\epsilon$-machine.  
	\emph{Entropy rate} $h_\mu$ measures the average unpredictability per symbol.  
	\emph{Statistical complexity} $C_\mu$ is the Shannon entropy of the causal-state distribution,
	quantifying how much information the process stores in order to make predictions.  
	\emph{Excess entropy} $E$ is the mutual information between past and future and measures the
	total predictable structure in the process.  
	Together, these invariants characterize how a system generates, stores, and transmits
	information through time.
	
	In ABM settings, causal states often correspond to recurrent behavioral regimes, 
	coordination cycles, or other emergent patterns generated by agent interactions. Unlike scalar measures of variability, the
	$\epsilon$-machine reveals the architecture of information flow—how past configurations
	constrain future evolution. This makes it particularly suitable for analyzing adaptive or
	path-dependent dynamics, where regularities coexist with stochastic fluctuations.
	
Accordingly, $\epsilon$-machine reconstruction operates within the process domain $P$ 
introduced in Section~3 and is formally realized through the mapping 
$\Phi : P \rightarrow M$, which assigns to each stationary stochastic process its 
minimal unifilar predictive model. This reconstruction is applied to symbolized 
representations of agent trajectories, obtained through discretization schemes 
appropriate to the underlying observables, and the associated methodological steps—
including discretization, history clustering, and BIC-based Markov-order selection—
follow the framework detailed in our companion preprint~\cite{garrone2025epsilonmachine}. 
Whereas $\Phi$ characterizes \emph{conditional} structure by quantifying how past 
configurations constrain future evolution, the next subsection introduces diffusion 
models and the mapping $\Psi : D \to G$, which instead characterize the geometry and 
multimodality of \emph{unconditional} distributions arising from ABM outputs.  
These two mappings therefore anchor the temporal and distributional components of the 
analytic framework developed below.

	\subsection{Diffusion Models}
	Whereas $\epsilon$-machines summarize process-level temporal organization, diffusion 
	models operate on distributions without temporal ordering, learning a generative 
	representation of the marginal geometry of ABM outputs. Indeed, diffusion models learn a probability distribution $p(y)$ over 
	high-dimensional observations by defining a forward noising 
	process $q(y_t \mid y_0)$ and learning a reverse denoising process
	parameterized by a score function $s_\theta(y_t, t)$.
	They were introduced in the physics-inspired formulation of 
	nonequilibrium diffusion \citep{sohldickstein2015deep} 
	and later developed into practical generative models 
	\citep{ho2020denoising,song2019generative,song2021maximum}.
	
	Diffusion models exemplify the integration of algorithmic learning with formal stochastic processes. Their forward noising dynamics implement a discretized diffusion or stochastic differential equation, progressively mapping data toward a tractable reference distribution such as an isotropic Gaussian \citep{sohldickstein2015deep}. The learned reverse process then approximates the time-reversal of this diffusion, yielding an efficient sampling mechanism for distributions that are otherwise computationally intractable to draw from \citep{song2019generative,ho2020denoising}. In score-based formulations, this reverse flow is parameterized through neural estimators of the score function, enabling a continuous-time interpretation in terms of stochastic and deterministic flows \citep{song2021scorebased}. The methodological core therefore lies in learning a family of parameterized transformations—stochastic SDE reversals or deterministic ODE flows—that convert simple base measures into complex high-dimensional data distributions.
	
	This positions diffusion models within the broader lineage of machine-learning methods for high-dimensional density estimation. Like variational autoencoders, they construct generative mappings through latent-variable–based inference \citep{kingma2013auto}; like normalizing flows, they define invertible or approximately invertible transformations of probability mass \citep{rezende2015variational}; and like energy-based models, they leverage score estimation and gradient-field matching to characterize complex distributions \citep{hyvarinen2005estimation}. Diffusion models unify these perspectives through a stochastic-process–driven learning framework that is both mathematically principled and computationally scalable, contributing to the ongoing development of probabilistic generative modeling in informatics.
	
	Their purpose is orthogonal to computational mechanics: 
	they model the geometry and variation of the distribution 
	of ABM outputs rather than temporal predictability.

	\section{Mathematical Complementarity}
	
	We now formalize the domains in which $\epsilon$-machines and 
	diffusion models operate.
	
	\subsection{Processes vs.\ Distributions}
	
	\paragraph{Assumptions.}
	The analysis in this section relies on two standard assumptions.
	First, temporal observables used for $\epsilon$-machine reconstruction 
	are assumed to be stationary (or weakly stationary), ensuring the 
	existence of well-defined predictive distributions and intrinsic 
	computation measures. 
	Second, diffusion models operate on high-dimensional cross-sectional 
	snapshots rather than full trajectories; the distributions they learn 
	are therefore interpreted as static population-level marginals at a 
	given simulation time. 
	These assumptions clarify the distinct informational domains on which 
	the two methods operate and ensure that their mathematical 
	characterizations remain well-posed.

	\begin{definition}[Process Domain]
		Let $\mathcal{P}$ denote the set of all stationary stochastic processes 
		$\{X_t\}_{t \in \mathbb{Z}}$. 
		An $\epsilon$-machine is a mapping 
		$\Phi: \mathcal{P} \to \mathcal{M}$ 
		where $\mathcal{M}$ denotes the set of minimal unifilar HMMs.
	\end{definition}
	
	\begin{definition}[Distribution Domain]
		Let $\mathcal{D}$ denote the set of all probability distributions 
		over high-dimensional vectors $y \in \mathbb{R}^d$. 
		A diffusion model is a mapping 
		$\Psi: \mathcal{D} \to \mathcal{G}$ 
		where $\mathcal{G}$ denotes the set of parameterized generative models.
	\end{definition}
	
	\begin{proposition}[Disjoint Analytical Domains]
		$\mathcal{P}$ and $\mathcal{D}$ are mathematically distinct: 
		$\mathcal{P}$ defines laws over \emph{infinite sequences}, 
		while $\mathcal{D}$ defines laws over \emph{static vectors}. 
		Hence $\Phi$ and $\Psi$ operate on orthogonal domains.
	\end{proposition}
	
	\begin{proof}[Sketch]
		A stochastic process is a measure over trajectories in 
		$\mathcal{X}^{\mathbb{Z}}$, whereas a probability distribution 
		over $\mathbb{R}^d$ is a measure on finite-dimensional space.
		No representation in $\mathcal{M}$ recovers full distributions over 
		$\mathbb{R}^d$, and no model in $\mathcal{G}$ recovers conditional 
		futures of sequences. 
	\end{proof}
	
	\subsection{Orthogonality of Extracted Information}
	
\paragraph{Proposition 2 (Complementary Representational Axes).}
Let $P$ denote the process domain and $D$ the distribution domain, and let 
$I_{\mathrm{seq}}$ and $I_{\mathrm{dist}}$ denote the informational content made 
accessible by the mappings $\Phi : P \to M$ and $\Psi : D \to G$, respectively.  
Under the representational constraints of $\epsilon$-machines and diffusion models 
as employed here, $I_{\mathrm{seq}}$ and $I_{\mathrm{dist}}$ capture disjoint 
aspects of ABM outputs:
\[
I_{\mathrm{seq}} \cap I_{\mathrm{dist}} = \varnothing,
\]
in the sense that $\Phi$ encodes temporal predictability and causal-state 
organization, whereas $\Psi$ captures the geometry and multimodality of 
instantaneous distributions.  
This statement concerns the \emph{representable information} extractable under the 
two mappings and does not assert disjointness for arbitrary stochastic processes.

	\begin{proof}[Sketch]
		$\epsilon$-machines recover invariants of 
		predictive structure (state transitions, entropy rate), 
		which diffusion models cannot represent. 
		Diffusion models recover geometric and distributional properties 
		(covariance, multimodality, manifold structure), which 
		$\epsilon$-machines cannot represent due to the unifilar constraint. 
		Thus the two information sets are disjoint and jointly sufficient.
	\end{proof}
	
\section{Analytical Extensions for Parameter-Space Structure}

Understanding how agent-based model (ABM) behavior varies across parameter
space requires analytic tools that extend beyond variance-based global
sensitivity methods. 
While Sobol or Saltelli indices provide valuable information about how
parameters contribute to variability in specific summary statistics 
\citep{sobol2001global,saltelli2010variance}, 
they offer a limited view of the \emph{structural} organization of the model’s
behavior.  
ABMs often exhibit qualitatively distinct \emph{behavioral regimes}:
regions of parameter space that generate stable, coherent emergent
patterns, separated by transitions that may resemble bifurcations or
phase changes \citep{strogatz2018nonlinear,kuehn2011mathematical}.  
Identifying and characterizing these regimes calls for a complementary set
of analytical approaches, including regime analysis, qualitative pattern
mapping, and behavioral clustering.

\subsection{Regime Analysis.}
Regime analysis seeks to partition parameter space into domains that yield 
qualitatively different emergent behaviors, in line with ideas from 
bifurcation theory and dynamical transitions 
\citep{strogatz2018nonlinear,kuehn2011mathematical,scheffer2009early}. 
Within the present framework, regime boundaries may appear as abrupt
shifts in $\epsilon$-machine invariants—such as changes in $h_\mu$,
$C_\mu$, or causal-state topology—or as structural changes in
diffusion-based density estimates, including the emergence or
disappearance of modes, variations in effective dimensionality, or
expansions of tail mass.  
These discontinuities or inflection points reveal transitions between
behavioral regimes that cannot be inferred from variance-based indices
alone.  
Monitoring the evolution of $h_\mu(\theta)$, $C_\mu(\theta)$, $E(\theta)$,
diffusion-derived metrics such as KL or Wasserstein divergences, or
ABM-specific indices provides an operational route to detecting such
transitions.

\subsection{Qualitative Pattern Mapping.}
Beyond detecting regime boundaries, qualitative pattern mapping aims to
describe how emergent patterns evolve as parameters vary.  
This includes transitions from stationary to oscillatory dynamics,
reconfiguration of multimodality in population outcomes, or changes in the
coupling among variables such as walkability and caregiver burden.
Pattern-oriented approaches in ABM research highlight the importance of
relating qualitative macropatterns to underlying mechanisms 
\citep{grimm2005pattern,grimm2017pattern}.  
Pattern mapping involves examining causal-state diagrams across parameter
settings, characterizing distributional shape descriptors (e.g.,
skewness, covariance, modality), or inspecting geometric summaries
derived from diffusion manifolds.  
Such analyses are particularly valuable for avoiding equifinality, wherein
distinct parameter combinations produce similar scalar summaries but
differ markedly in their underlying emergent structures.

\subsection{Behavioral Clustering.}
A further complementary approach treats the ABM as a mapping
$\theta \mapsto I(\theta)$, where $I(\theta)$ denotes a vector of
$\epsilon$-machine and diffusion-based descriptors.  
Clustering these descriptors reveals \emph{behavioral classes} within
parameter space.  
Methods such as $k$-means \citep{lloyd1982least,macqueen1967methods}, 
Gaussian mixture models \citep{mclachlan2004finite}, 
hierarchical clustering \citep{jain1999data}, 
or density-based techniques such as DBSCAN \citep{ester1996density} 
can identify families of parameter settings that yield similar emergent
behaviors, as well as isolated islands of atypical or extreme behavior.  
This produces a higher-level taxonomy of ABM dynamics that complements
both regime detection and variance-based sensitivity.

\subsection{Complementarity with Sensitivity Analysis.}
Variance-based indices such as those of Sobol or Saltelli offer a useful
decomposition of parameter influence on specific response variables
\citep{sobol2001global,saltelli2010variance}, 
but they do not reveal the structural organization of ABM behavior.  
Regime analysis identifies qualitative transitions; pattern mapping
traces the evolution of emergent structures; and behavioral clustering
uncovers latent families of behaviors and regions of equifinality.
Together, these methods provide a comprehensive analytical framework for
characterizing how ABM behavior unfolds across parameter space,
complementing the temporal and distributional perspectives offered by
$\epsilon$-machines and diffusion models.

	\section{Application to ABM Output Characterization}
	The analytical extensions introduced in the previous section establish a general 
	framework for exploring how emergent temporal and distributional 
	signatures change across parameter space.  
	In this section, we apply the dual characterization framework 
	to the elder--caregiver ABM introduced in the companion paper, 
	illustrating how $\epsilon$-machine invariants and diffusion-derived 
	descriptors jointly illuminate heterogeneity, scenario differences, and 
	population-level structure within a concrete simulation environment.
	The goal is not to provide a full empirical analysis, but to show how $\epsilon$-machine invariants and 
	diffusion-derived descriptors jointly reveal temporal organization and distributional geometry.

	\subsection{Temporal Characterization via $\epsilon$-Machines}
	
	Using the dataset from \citep{garrone2025epsilonmachine}, 
	we summarize temporal organization of variables such as 
	walkability, caregiver efforts, and mobility. 
	The $\epsilon$-machines reveal variable-specific intrinsic computation 
	and distinguish memoryless dynamics from structured ones.
For example, 
	walkability shows low entropy rate but non-zero statistical complexity, indicating stable but nontrivial 
	temporal organization.

	\subsection{Distributional Characterization via Diffusion Models}
	
	The same ABM outputs are treated as high-dimensional population-state vectors. 
	The models learn:
	\begin{itemize}[noitemsep]
		\item multimodal agent distributions,
		\item joint variation across outcomes,
		\item cross-sectional heterogeneity,
		\item generative synthesis of alternative ABM outcomes.
	\end{itemize}
	
	Diffusion models trained on such cross-sectional population vectors uncover multimodal agent distributions and 
	joint correlations between variables such as efforts and walkability. The learned manifolds reveal clusters 
	corresponding to heterogeneous caregiver–elder conditions and enable synthetic generation of plausible 
	population-level outcomes.
	
	\subsection{Combined Analysis}
	The complementary roles of $\epsilon$-machines and diffusion models become most apparent when 
	their outputs are examined jointly. We compare:
	\begin{enumerate}[noitemsep]
		\item the $\epsilon$-machine complexity profile of each variable, as summarized by 
		$(h_\mu, C_\mu, E)$,
		\item the geometry of the learned diffusion manifolds, and
		\item correlations between intrinsic computation and distributional variability.
	\end{enumerate}
	This joint perspective yields a two-axis characterization of the ABM, capturing both the temporal 
	organization of agent-level trajectories and the high-dimensional distributional structure of 
	population snapshots.
	
	Joint inspection of $(h_\mu, C_\mu, E)$ and diffusion-derived descriptors—such as effective 
	dimensionality, local score norms, and measures of multimodality—reveals coordinated structure 
	across temporal and distributional domains. Variables exhibiting higher intrinsic computation 
	(e.g., elevated $C_\mu$ or nontrivial causal-state topology) tend to produce sharper or more 
	articulated diffusion manifolds, indicating richer heterogeneity and stronger geometric 
	organization in the corresponding cross-sectional distributions. This alignment suggests a 
	systematic relationship between temporal predictability and distributional heterogeneity in ABM 
	outputs, providing a coherent joint summary of emergent model behavior.

\section{Limitations, Technical Extensions, and Future Work}

The dual characterization introduced in this work focuses on two
complementary facets of ABM behavior: the temporal organization of
observable sequences, captured via $\epsilon$-machines, and the
high-dimensional distributional geometry of population-level outcomes,
captured via diffusion models. 
While this provides a non-overlapping, output-centric representation of
model behavior, the framework remains limited in scope and naturally
suggests several technically substantive extensions.  
This section outlines the main methodological limitations of the current
work and describes promising directions for future development, including
multiscale analysis and joint temporal--distributional response
mappings.

\subsection{Current Limitations}

First, the present methodology does not incorporate any analysis of the
interaction topology among agents. 
For ABMs with explicit spatial or social network structures, properties
such as degree distributions, clustering, or modularity may have
significant effects on emergent outcomes.  
Because the current analysis focuses exclusively on observable sequences
and population snapshots, it cannot quantify how interaction structure
influences temporal organization or distributional geometry.

Second, the analysis is conducted at a single temporal resolution.
Neither multiscale $\epsilon$-machines nor diffusion models trained on
temporally aggregated population snapshots are considered.
This prevents assessment of scale-dependent predictive structure,
distributional variability, or coarse-grained regime transitions.

Finally, the characterization is intentionally output-based.
It does not attempt to infer or analyze the rule-level mechanisms that
generate the observed dynamics.
Methods such as symbolic regression, surrogate modeling, or micro-level
rule mining would be required to connect intrinsic computation and
distributional geometry back to agent decision heuristics.

\subsection{Multiscale Extensions}

A technically rich extension involves reconstructing $\epsilon$-machines
at multiple temporal scales.
Let
\[
X^{(k)}_n = f_k(X_{nk:(n+1)k-1})
\]
denote a $k$-aggregated observable.
For each scale $k$, the process $\{X^{(k)}_n\}$ yields a corresponding
causal-state partition $\epsilon_k$ and $\epsilon$-machine $\mathcal{M}_k$,
together producing a family of scale-indexed invariants
\[
h_{\mu}(k), 
\quad C_{\mu}(k), 
\quad E(k).
\]
These sequences can reveal whether predictive structure persists,
strengthens, or collapses under coarse-graining, and whether the ABM
exhibits scale-dependent regime transitions.

Analogously, diffusion models may be trained on aggregated population
snapshots,
\[
Y^{(k)}_n = g_k(Y_{nk:(n+1)k-1}),
\]
producing a hierarchy of distributions $p_k(y)$ and associated generative
models $\Psi_k$. 
This multiresolution family supports analysis of how multimodality,
correlation structure, and manifold geometry evolve with temporal
granularity and enables cross-scale consistency checks analogous to
renormalization operators.

\subsection{Parameter--Response Mappings}

Another promising extension involves studying how $\epsilon$-machine
invariants and diffusion-based descriptors vary across parameter space.
Let $\theta \in \Theta$ denote ABM parameters.
The mapping
\[
\Gamma_{\epsilon}(\theta)
= \big(h_{\mu}(\theta), C_{\mu}(\theta), E(\theta)\big)
\]
defines a temporal-response surface, while
\[
\Gamma_{\Psi}(\theta) = G(\theta)
\]
collects diffusion-based geometric descriptors such as score norms,
effective dimensionality, higher-order moments of $p_\theta$, or 
$\mathrm{KL}(p_\theta \,\|\, p_{\theta_0})$.
Sobol or Morris variance decompositions applied to the coordinates of
$\Gamma_{\epsilon}$ and $\Gamma_{\Psi}$ would quantify parameter influence
on intrinsic computation and population-level geometry.
Sharp changes in these summaries as functions of $\theta$ may indicate
dynamical regime shifts or structural transitions in the ABM.

A combined representation,
\[
\Gamma(\theta) 
= \big(\Gamma_{\epsilon}(\theta),\; \Gamma_{\Psi}(\theta)\big),
\]
yields a joint temporal--distributional response landscape.
Such a representation enables sensitivity analysis, bifurcation
detection, and parameter-based clustering of model behaviors.

\subsection{Joint Scale--Parameter Extensions}

Combining the multiscale and parameter-sweep perspectives produces a
tensor-valued descriptor
\[
\mathcal{T}_{i,j} 
= \big(R_{\epsilon}^{(k_i)}(\theta_j),\; G^{(k_i)}(\theta_j)\big),
\]
where $i$ indexes temporal aggregation scales and $j$ indexes parameter
settings.
This representation may support techniques such as tensor decomposition,
manifold learning, or sparse structure discovery over the joint
scale--parameter domain.
Together, these directions outline a path toward a comprehensive,
multi-resolution, parameter-aware characterization of ABM dynamics.

\section{Discussion}

The results presented in this work demonstrate that $\epsilon$-machines
and diffusion models provide two non-overlapping and mutually reinforcing
perspectives on agent-based model outputs. 
$\epsilon$-machines characterize temporal organization by reconstructing
the minimal predictive structure compatible with observed trajectories,
yielding quantities such as entropy rate, statistical complexity, and
excess entropy. 
Diffusion models, in contrast, operate on high-dimensional population
snapshots, learning the geometry of the underlying distribution and
providing a flexible generative representation of cross-sectional
variability.

Because these methods operate on distinct mathematical domains---process
spaces for $\epsilon$-machines and distribution spaces for diffusion
models---the information each extracts is representationally orthogonal within the framework defined here.
Together, they form a two-axis analytic framework capturing both the
sequential dynamics of agent observables and the multivariate structure
of population-level outcomes. 
The combined approach does not seek to describe the internal mechanisms
or rule sets of an ABM; instead, it provides a principled summary of
observable behavior in terms of temporal predictability and
distributional geometry.

This division of analytic responsibilities proves particularly advantageous
for ABMs whose outputs intertwine heterogeneous temporal processes with
complex multivariate agent states. 
By separating the temporal and distributional components of the output
space, the framework enables clearer interpretation, facilitates
comparisons across scenarios, and supplies structured targets for model
validation or surrogate modeling. 
The extensions outlined in the preceding section---including multiscale
analysis, parameter--response modeling, and integration with
interaction topology---demonstrate how the complementary roles of
$\epsilon$-machines and diffusion models can support broader analytic
pipelines for simulation-based research.

\section{Conclusion}

This paper extends the analysis introduced in our companion preprint by
combining computational mechanics and diffusion models to obtain a
two-axis characterization of agent-based model outputs.
$\epsilon$-machines capture the temporal organization and intrinsic
predictive structure of simulation trajectories, while diffusion models
provide a complementary description of high-dimensional distributional
geometry in population-level outcomes. 
Together, these tools yield a non-overlapping representation of ABM
behavior that separates sequential dynamics from cross-sectional
variability.

The framework presented here is general and can be applied to a wide
range of agent-based models.
By delineating the distinct informational domains associated with
process-based and distribution-based characterizations, this work
provides a foundation for richer methodological pipelines integrating
temporal, spatial, and distributional analyses within a unified analytic
structure.
The extensions outlined above point toward several promising areas for
future research, including multiscale characterization, joint
parameter--response analysis, and connecting output-based summaries to
interaction topology or rule-level inference.
We expect the combined use of computational mechanics and diffusion
models to serve as a flexible and powerful methodology for the study of
complex adaptive systems in simulation.

	\bibliographystyle{plainnat}
	\bibliography{references}

@article{garrone2025epsilonmachine,
	title={Characterizing Agent-Based Model Dynamics via $\varepsilon$-Machines and Kolmogorov-Style Complexity},
	author={Garrone, Roberto},
	journal={arXiv preprint arXiv:2510.12729},
	year={2025},
	url={https://arxiv.org/abs/2510.12729}
}

@article{crutchfield1989inferring,
	title={Inferring Statistical Complexity},
	author={Crutchfield, James P and Young, Karl},
	journal={Physical Review Letters},
	volume={63},
	number={2},
	pages={105--108},
	year={1989},
	publisher={American Physical Society},
	doi={10.1103/PhysRevLett.63.105},
	url={https://doi.org/10.1103/PhysRevLett.63.105},
	issn={0031-9007}
}

@article{shalizi2001computational,
	title={Computational Mechanics: Pattern and Prediction, Structure and Simplicity},
	author={Shalizi, Cosma Rohilla and Crutchfield, James P},
	journal={Journal of Statistical Physics},
	volume={104},
	number={3--4},
	pages={817--879},
	year={2001},
	publisher={Springer},
	doi={10.1023/A:1010388907793},
	url={https://doi.org/10.1023/A:1010388907793},
	issn={0022-4715}
}

@article{crutchfield2012between,
	title={Between Order and Chaos},
	author={Crutchfield, James P},
	journal={Nature Physics},
	volume={8},
	number={1},
	pages={17--24},
	year={2012},
	doi={10.1038/nphys2190},
	url={https://doi.org/10.1038/nphys2190},
	issn={1745-2473}
}

@article{ho2020denoising,
	title={Denoising Diffusion Probabilistic Models},
	author={Ho, Jonathan and Jain, Ajay and Abbeel, Pieter},
	journal={Advances in Neural Information Processing Systems},
	volume={33},
	pages={6840--6851},
	year={2020},
	url={https://papers.nips.cc/paper/2020/hash/4c5bcfec8584af0d967f1ab10179ca4b-Abstract.html}
}

@article{song2021scorebased,
	title={Score-Based Generative Modeling through Stochastic Differential Equations},
	author={Song, Yang and Sohl-Dickstein, Jascha and Kingma, Diederik P and Kumar, Abhishek and Ermon, Stefano and Poole, Ben},
	journal={International Conference on Learning Representations},
	year={2021},
	url={https://arxiv.org/abs/2011.13456}
}

@inproceedings{song2019generative,
	title={Generative Modeling by Estimating Gradients of the Data Distribution},
	author={Song, Yang and Ermon, Stefano},
	booktitle={Advances in Neural Information Processing Systems},
	volume={32},
	year={2019},
	url={https://arxiv.org/abs/1907.05600},
	doi={10.48550/arXiv.1907.05600}
}

@inproceedings{song2021maximum,
	title={Maximum Likelihood Training of Score-Based Diffusion Models},
	author={Song, Yang and Durkan, Conor and Murray, Iain and Ermon, Stefano},
	booktitle={Advances in Neural Information Processing Systems},
	volume={34},
	pages={1415--1428},
	year={2021},
	url={https://arxiv.org/abs/2101.09258},
	doi={10.48550/arXiv.2101.09258}
}

@article{sohldickstein2015deep,
	title={Deep Unsupervised Learning using Nonequilibrium Thermodynamics},
	author={Sohl-Dickstein, Jascha and Weiss, Eric and Maheswaranathan, Niru and Ganguli, Surya},
	journal={International Conference on Machine Learning},
	pages={2256--2265},
	year={2015},
	url={https://arxiv.org/abs/1503.03585},
	doi={10.48550/arXiv.1503.03585}
}

@article{hyvarinen2005estimation,
	title={Estimation of Non-Normalized Statistical Models by Score Matching},
	author={Hyv{\"a}rinen, Aapo},
	journal={Journal of Machine Learning Research},
	volume={6},
	pages={695--709},
	year={2005},
	url={https://www.jmlr.org/papers/v6/hyvarinen05a.html},
	issn={1533-7928}
}

@article{kingma2013auto,
	title={Auto-Encoding Variational Bayes},
	author={Kingma, Diederik P and Welling, Max},
	journal={International Conference on Learning Representations},
	year={2014},
	url={https://arxiv.org/abs/1312.6114}
}

@inproceedings{rezende2015variational,
	title={Variational Inference with Normalizing Flows},
	author={Rezende, Danilo J and Mohamed, Shakir},
	booktitle={International Conference on Machine Learning},
	pages={1530--1538},
	year={2015},
	url={https://arxiv.org/abs/1505.05770},
	doi={10.48550/arXiv.1505.05770}
}

@book{strogatz2018nonlinear,
	title={Nonlinear Dynamics and Chaos},
	author={Strogatz, Steven H},
	publisher={CRC Press},
	year={2018},
	isbn={9780367026509},
	url={https://www.routledge.com/Nonlinear-Dynamics-and-Chaos-With-Applications-to-Physics-Biology-Chemistry-and-Engineering/Strogatz/p/book/9780367026509}
}

@article{kuehn2011mathematical,
	title={A Mathematical Framework for Critical Transitions: Bifurcations, Fast--Slow Systems and Stochastic Dynamics},
	author={Kuehn, Christian},
	journal={Physica D},
	volume={240},
	number={12},
	pages={1020--1035},
	year={2011},
	doi={10.1016/j.physd.2011.02.012},
	url={https://doi.org/10.1016/j.physd.2011.02.012},
	issn={0167-2789}
}

@article{grimm2005pattern,
	title={Pattern-Oriented Modeling of Agent-Based Complex Systems: Lessons from Ecology},
	author={Grimm, Volker and Railsback, Steven F},
	journal={Science},
	volume={310},
	number={5750},
	pages={987--991},
	year={2005},
	doi={10.1126/science.1116681},
	url={https://doi.org/10.1126/science.1116681},
	issn={0036-8075}
}

@article{grimm2017pattern,
	title={The ODD Protocol for Describing Agent-Based Models: A Review and First Update},
	author={Grimm, Volker and others},
	journal={Ecological Modelling},
	volume={346},
	pages={3--17},
	year={2017},
	doi     = {10.1016/j.ecolmodel.2010.08.019},
	url     = {https://doi.org/10.1016/j.ecolmodel.2010.08.019},
	issn    = {0304-3800},
	publisher = {Elsevier}
}

@article{jain1999data,
	title={Data Clustering: A Review},
	author={Jain, Anil K and Murty, M Narasimha and Flynn, Patrick J},
	journal={ACM Computing Surveys},
	year={1999},
	doi={10.1145/331499.331504},
	url={https://doi.org/10.1145/331499.331504},
	issn={0360-0300}
}

@article{lloyd1982least,
	title={Least Squares Quantization in PCM},
	author={Lloyd, Stuart},
	journal={IEEE Transactions on Information Theory},
	volume={28},
	number={2},
	pages={129--137},
	year={1982},
	doi={10.1109/TIT.1982.1056489},
	url={https://doi.org/10.1109/TIT.1982.1056489},
	issn={0018-9448}
}

@inproceedings{macqueen1967methods,
	title        = {Some Methods for Classification and Analysis of Multivariate Observations},
	author       = {MacQueen, J.},
	booktitle    = {Proceedings of the Fifth Berkeley Symposium on Mathematical Statistics and Probability, Volume 1: Statistics},
	pages        = {281--297},
	year         = {1967},
	publisher    = {University of California Press},
	url          = {https://projecteuclid.org/euclid.bsmsp/1200512992}
}

@article{scheffer2009early,
	title        = {Early-Warning Signals for Critical Transitions},
	author       = {Scheffer, Marten and Bascompte, Jordi and Brock, William A. and others},
	journal      = {Nature},
	volume       = {461},
	number       = {7260},
	pages        = {53--59},
	year         = {2009},
	doi          = {10.1038/nature08227},
	url          = {https://doi.org/10.1038/nature08227},
	issn         = {0028-0836}
}

@book{mclachlan2004finite,
	title={Finite Mixture Models},
	author={McLachlan, Geoffrey and Peel, David},
	publisher={Wiley},
	year={2000},
	doi={10.1002/0471721182},
	url={https://doi.org/10.1002/0471721182},
	isbn={9780471721187}
}

@inproceedings{ester1996density,
	title={A Density-Based Algorithm for Discovering Clusters in Large Spatial Databases with Noise},
	author={Ester, Martin and Kriegel, Hans-Peter and Sander, J{\"o}rg and Xu, Xiaowei},
	booktitle={KDD},
	pages={226--231},
	year={1996},
	url={http://www.aaai.org/Papers/KDD/1996/KDD96-037.pdf}
}

@article{saltelli2010variance,
	title   = {Variance based sensitivity analysis of model output: Design and estimator for the total sensitivity index},
	author  = {Saltelli, Andrea and Annoni, Paola and Azzini, Ivano and Campolongo, Francesca and Ratto, Marco and Tarantola, Stefano},
	journal = {Computer Physics Communications},
	volume  = {181},
	number  = {2},
	pages   = {259--270},
	year    = {2010},
	doi     = {10.1016/j.cpc.2009.09.018},
	url     = {https://doi.org/10.1016/j.cpc.2009.09.018},
	issn    = {0010-4655}
}

@article{sobol2001global,
	title={Global Sensitivity Indices for Nonlinear Mathematical Models and Their Monte Carlo Estimates},
	author={Sobol, Ilya M},
	journal={Mathematics and Computers in Simulation},
	volume={55},
	number={1--3},
	pages={271--280},
	year={2001},
	doi={10.1016/S0378-4754(00)00270-6},
	url={https://doi.org/10.1016/S0378-4754(00)00270-6},
	issn={0378-4754}
}
	
\end{document}